\documentclass{article}

\usepackage[numbers, sort, compress]{natbib}
\makeatletter
\usepackage[verbose=true,letterpaper]{geometry}
\AtBeginDocument{
  \newgeometry{
    textheight=9in,
    textwidth=5.5in,
    top=1in,
    headheight=12pt,
    headsep=25pt,
    footskip=30pt
  }
}

\usepackage[utf8]{inputenc} % allow utf-8 input
\usepackage[T1]{fontenc}    % use 8-bit T1 fonts
\usepackage{hyperref}       % hyperlinks
\usepackage{url}            % simple URL typesetting
\usepackage{booktabs}       % professional-quality tables
\usepackage{amsfonts}       % blackboard math symbols
\usepackage{nicefrac}       % compact symbols for 1/2, etc.
\usepackage{microtype}      % microtypography
\usepackage{xcolor}         % colors

\usepackage{mathtools}
\usepackage{amsthm}
\usepackage{amssymb}
\usepackage{tikzit}
\usepackage{tikz-cd}
\usepackage{mathtools}
\usepackage{quantikz}
\usepackage{adjustbox}
\usepackage{multirow}
\usepackage[nameinlink, capitalise]{cleveref}

\DeclarePairedDelimiter\abs{\lvert}{\rvert}%
\DeclarePairedDelimiter\norm{\lVert}{\rVert}%

\tikzstyle{gate}=[shape=rectangle, text height=1.5ex, text depth=0.25ex, yshift=0.5mm, fill=white, draw=black, minimum height=5mm, yshift=-0.5mm, minimum width=5mm, font={\small}]
\tikzstyle{big gate}=[shape=rectangle, text height=1.5ex, text depth=0.25ex, yshift=0.5mm, fill=white, draw=black, minimum height=10mm, yshift=-0.5mm, minimum width=5mm, font={\small}]
\tikzstyle{Z dot}=[inner sep=0mm, minimum size=2mm, shape=circle, draw=black, fill={rgb,255: red,221; green,255; blue,221}]
\tikzstyle{Z phase dot}=[minimum size=5mm, font={\footnotesize\boldmath}, shape=rectangle, rounded corners=2mm, inner sep=0.4mm, outer sep=-2mm, scale=0.8,  draw=black, fill={rgb,255: red,221; green,255; blue,221}]
\tikzstyle{X dot}=[Z dot, shape=circle, draw=black, fill={rgb,255: red,255; green,136; blue,136}]
\tikzstyle{X phase dot}=[Z phase dot, fill={rgb,255: red,255; green,136; blue,136}, font={\footnotesize\boldmath}]
\tikzstyle{hadamard}=[fill=yellow, draw=black, shape=rectangle, inner sep=0.6mm, minimum height=1.5mm, minimum width=1.5mm]
\tikzstyle{paulibox}=[fill={rgb,255: red,221; green,221; blue,255}, draw=black, shape=rectangle, inner sep=0.6mm, minimum height=5mm, minimum width=5mm, font={\footnotesize}, text height=1.5ex, text depth=0.25ex]
\tikzstyle{vertex}=[inner sep=0mm, minimum size=1mm, shape=circle, draw=black, fill=black]
\tikzstyle{vertex set}=[inner sep=0mm, minimum size=1mm, shape=circle, draw=black, fill=white, font={\footnotesize\boldmath}]
\tikzstyle{small black dot}=[fill=black, draw=black, shape=circle, inner sep=0pt, minimum width=1.2mm]
\tikzstyle{cnot ctrl}=[fill=black, draw=black, shape=circle, inner sep=0pt, minimum width=1.2mm]
\tikzstyle{cnot targ}=[fill=white, draw=white, shape=circle, label={center:$\oplus$}, inner sep=0pt, minimum width=2.1mm]
\tikzstyle{ket}=[fill=white, draw=black, shape=regular polygon, regular polygon sides=3, regular polygon rotate=-30, scale=0.7, inner sep=1pt]
\tikzstyle{bra}=[fill=white, draw=black, shape=regular polygon, regular polygon sides=3, regular polygon rotate=30, scale=0.7, inner sep=1pt]
\tikzstyle{scalar}=[shape=rectangle, text height=1.5ex, text depth=0.25ex, yshift=0.5mm, fill=white, draw=black, minimum height=5mm, yshift=-0.5mm, minimum width=5mm, font={\small}]
\tikzstyle{clabel}=[fill=white, draw=none, shape=rectangle, inner sep=1pt]
\tikzstyle{empty diagram}=[draw={gray!40!white}, dashed, shape=rectangle, minimum width=1cm, minimum height=1cm]
\tikzstyle{amap}=[fill=white, draw=black, shape=NEbox]
\tikzstyle{amap conj}=[fill=white, draw=black, shape=NWbox]
\tikzstyle{amap adj}=[fill=white, draw=black, shape=SEbox]
\tikzstyle{amap trans}=[fill=white, draw=black, shape=SWbox]
\tikzstyle{astate}=[fill=white, draw=black,line width=1.6pt, shape=NEtriangle]
\tikzstyle{astate conj}=[fill=white, draw=black, line width=1.6pt, shape=NWtriangle]
\tikzstyle{astate adj}=[fill=white, draw=black, line width=1.6pt, shape=SEtriangle]
\tikzstyle{astate trans}=[fill=white, draw=black, line width=1.6pt, shape=SWtriangle]

\tikzstyle{hadamard edge}=[-, dashed, dash pattern=on 2pt off 0.5pt, thick, draw={rgb,255: red,68; green,136; blue,255}]
\tikzstyle{box edge}=[-, dashed, dash pattern=on 2pt off 0.5pt, thick, draw={rgb,255: red,203; green,192; blue,225}]
\tikzstyle{gray dashed edge}=[-, dashed, dash pattern=on 2pt off 0.5pt, thick, {gray!60!white}]
\tikzstyle{brace edge}=[-,  decorate, decoration={brace,amplitude=1mm,raise=-1mm}]
\tikzstyle{diredge}=[->]
\tikzstyle{double edge}=[-, double, shorten <=-1mm, shorten >=-1mm, double distance=2pt]
\tikzstyle{gray edge}=[-, {gray!60!white}]
\tikzstyle{blue edge}=[-, {blue!60!white}]
\tikzstyle{pointer edge}=[->, very thick, gray]
\tikzstyle{boldedge}=[-, line width=1.6pt, shorten <=-0.17mm, shorten >=-0.17mm]
\tikzstyle{bidir edge}=[<->, very thick, draw={rgb,255: red,191; green,191; blue,191}]
\tikzstyle{new edge style 0}=[-, fill=white]
\pgfdeclarelayer{edgelayer}
\pgfdeclarelayer{nodelayer}
\pgfsetlayers{edgelayer,nodelayer,main}
\tikzstyle{none}=[inner sep=0pt]

\renewcommand{\ket}[1]{\ensuremath{|#1\rangle}}
\renewcommand{\bra}[1]{\ensuremath{\langle #1|}}
\renewcommand{\braket}[2]{\ensuremath{\langle #1|#2\rangle}}
\newcommand{\ketbra}[2]{\ensuremath{|#1\rangle\!\langle #2|}}

\newtheorem{lemma}{Lemma}

\title{Quixer: A Quantum Transformer Model}

\author{%
  Nikhil Khatri  \quad
  Gabriel Matos \quad
  Luuk Coopmans \quad
  Stephen Clark \\\\
  Quantinuum \\
  17 Beaumont St., Oxford OX1 2NA, UK \\
  \small{\texttt{\{nikhil.khatri,gabriel.matos,luuk.coopmans,steve.clark\}@quantinuum.com}} \\
}

\newcommand{\uprep}{U_{\text{PREP}}}

\newcommand{\usel}{U_{\text{SEL}}}

\begin{document}

\maketitle

\begin{abstract}
Progress in the realisation of reliable large-scale quantum computers has motivated research into the design of quantum machine learning models.
We present Quixer: a novel quantum transformer model which utilises the Linear Combination of Unitaries and Quantum Singular Value Transform primitives as building blocks. Quixer operates by preparing a superposition of tokens and applying a trainable non-linear transformation to this mix. We present the first results for a quantum transformer model applied to a practical language modelling task, obtaining results competitive with an equivalent classical baseline. In addition, we include resource estimates for evaluating the model on quantum hardware, and provide an open-source implementation for classical simulation. We conclude by highlighting the generality of Quixer, showing that its parameterised components can be substituted with fixed structures to yield new classes of quantum transformers.
\end{abstract}

\section{Introduction}
\label{sec:introduction}

Remarkable developments have been achieved in natural language processing, leading to the advent and popularisation of large language models (LLMs) \cite{touvron2023llama, gemmateam2024gemma, jiang2023mistral}. At the same time, significant progress has been made in the field of quantum computing. While current quantum devices are still noisy~\cite{Preskill2018}, rapid improvements \cite{dasilva2024demonstration, delaney2024scalable} are driving the field into the error-corrected, fault-tolerant regime, where algorithms with asymptotic speed-up over their classical counterparts can be successfully run~\cite{Montanaro2016}.

While powerful, LLMs are notoriously costly to train due to the number of parameters in state-of-the-art architectures~\cite{kaplan2020scaling}. Thus, it is of great practical interest to find alternative efficient, yet performant models. Given that quantum computers are known to provide a complexity-theoretic advantage in certain domains~\cite{Shor1997, Grover1996}, it is natural to explore what a quantum version of the transformer architecture could look like. While the original proposal for the Transformer model~\cite{vaswani2017} uses a dot product self-attention mechanism, other architectures employ alternatives which are nonetheless performant e.g. the FNet~\cite{lee2022fnet}.

In this work, we propose Quixer (QUantum mIXER), a quantum transformer model that incorporates quantum algorithmic primitives in a novel attention mechanism. A Linear Combination of Unitaries (LCU)~\cite{childs2012} procedure is employed to create a superposition of token unitaries, and a Quantum Singular Value Transform (QSVT)~\cite{Gilyn2019} primitive is used to further apply a non-linear transformation to this superposition. Our primary contributions are the first result for a quantum transformer model applied to a practical language modelling task, along with a novel quantum attention mechanism built using the LCU and QSVT.

\Cref{fig:quixer} presents the high-level components of the Quixer model, which we describe in detail in \cref{sec:model}. \Cref{sec:background} provides the necessary background information on classical transformers and quantum computing. In \cref{sec:results}, we present results for a language modelling task on the Penn Treebank dataset, comparing the performance of Quixer against various contemporary classical models. Finally, in \cref{sec:framework}, we outline the extensible nature of Quixer, and highlight directions for future research.

\begin{figure}[ht]
    \centering
    \includegraphics[scale=0.70]{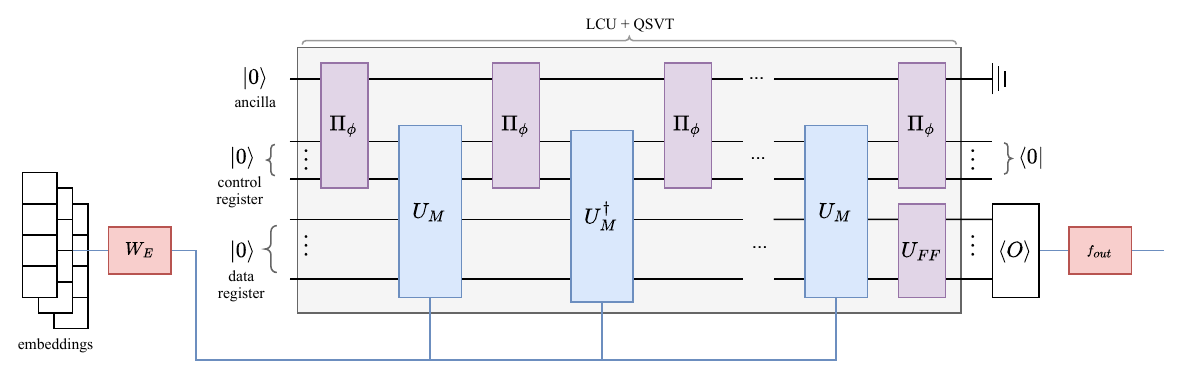}
    \caption{Quixer model architecture. Token embeddings are linearly mapped to angles by $W_E$. These angles parameterise unitary quantum circuits acting on a data register. A linear superposition of these unitaries is prepared using an LCU $U_M$, to which a polynomial, specified through the phases in the $\Pi_\phi$ gates, is applied through a QSVT procedure. A feed-forward unitary $U_{\textit{FF}}$ is then applied to the data register. Data is read out from the resulting quantum state using multiple measurement operators, and the resulting expectation values $\langle O \rangle$ are classically processed by $f_{\textit{out}}$ to produce the final output of the model.}
    \label{fig:quixer}
\end{figure}

\subsection{Related work}
\label{sec:related_work}

Other types of quantum models have been proposed for natural language processing, such as quantum recurrent neural networks (QRNN)~\cite{bausch_qrnn_2020, widdows2024natural}. In this section, we focus on literature concerning quantum models which are related either to transformers or language modelling. We refer the reader to \citet{widdows2024natural} for a more general overview of quantum natural language processing. To the best of our knowledge, other than the present work, the only previous results for language modelling using a quantum architecture have been presented by \citet{basile2017}, which uses a QRNN-like model to tackle a phone-recognition task on the TIMIT corpus.

Recently, there have been several proposals for quantum versions of the transformer architecture. \citet{Cherrat2024} tackle a classification task using a family of MNIST datasets, and employ methods which differ significantly from ours. In particular, the authors use a data loading method to directly encode a trainable matrix on the amplitudes of a quantum state, and process this using a specialised ``orthogonal layer'' circuit. In contrast, \citet{liao2024gpt} and \citet{guo2024quantum} attempt to directly quantise each component of the transformer architecture in a modular fashion. Their work is theoretical in nature and does not present results on a concrete language task.

Another proposal for a quantum transformer architecture is given by \citet{zhao2024gqhan}, which directly incorporates a mechanism akin to Grover's algorithm, and is trained on the Fashion MNIST dataset. Likewise, \citet{gao2023fast} use a Grover-like mechanism for sparse attention, though their work is fully theoretical. Finally, \citet{zhao2023qksan} implement a transformer based on a quantum kernel method, and perform binary classification on MNIST and Fashion MNIST.

\section{Background}
\label{sec:background}

\subsection{The Transformer architecture} \label{sec:transformer}
The heart of a Transformer is the multi-head dot product self-attention mechanism introduced by Vaswani et al. \cite{vaswani2017}. Several variations on the Transformer have been proposed in the literature, which replace the dot product self-attention mechanism with alternate methods to mix information between tokens. For instance, some substitute the quadratic-time dot product self-attention mechanisms with linear time attention mechanisms \cite{katharopoulos2020transformers,wang2020linformer}. Other variations exist which replace the attention unit with entirely untrainable transformations, such as the Fourier transform \cite{lee2022fnet}. These variants demonstrate that the specific dot product self-attention mechanism is not necessary to make a performant transformer. Motivated by the performance of these alternatives to the original transformer self-attention mechanism, our focus in this work is not to quantise the dot product self-attention, but instead propose a novel form of token mixing built from quantum primitives.

\subsection{Quantum computation} \label{sec:quantum}
In circuit-based quantum computing, a number $q$ of \emph{qubits} is manipulated. The \emph{state} $\ket{\psi} \in \mathbb{C}^2$ of each qubit is a normalised superposition (complex linear combination) of two \emph{computational basis} states
\begin{align}
    \ket{\psi} = \alpha \ket{0} + \beta \ket{1},\qquad
    \ket{0} = \begin{bmatrix}
        1\\0
    \end{bmatrix},\quad
    \ket{1} = \begin{bmatrix}
        0\\1
    \end{bmatrix},\qquad
    \abs{\alpha}^2 + \abs{\beta}^2 = 1, \quad \alpha, \beta \in \mathbb{C}.
\end{align}
The joint state of the qubits is an element of the tensor product $\bigotimes_{j=1}^q \mathbb{C}^2$, which is isomorphic to $\mathbb{C}^{2^q}$ as a vector space; the computational basis of this space is $\{\ket{c_{q-1}...c_0}: c_k \in \{0, 1\}\}$. For a $q$-qubit state and a positive integer $j=\sum_{k=0}^{q-1} c_k 2^k, c_k \in \{0,1\}$, we define $\ket{j} := \ket{c_{q-1}...c_0}$ (i.e. the $k$'th qubit is set to $\ket{c_k}$; for instance, $\ket{6} = \ket{110}$).
Quantum circuits are often drawn in diagrammatic form as in e.g. \eqref{fig:lcu}. In these diagrams, wires (represented by horizontal lines) can either be single qubits, or a group of qubits (i.e. a \emph{register}) when crossed by an oblique line. 

The state of a quantum system evolves through the application of \emph{gates}, each of which is represented by a unitary matrix $U$ (i.e. a matrix satisfying $U^{\dagger}U = I$, where $U^\dagger := \overline{U}^T$). Gates are drawn as boxes in circuit diagrams, and may act on single qubits or be \emph{entangling}, i.e. act on multiple qubits. An example of a single-qubit gate is the $R_{X}$ gate, defined as $R_{X}(\theta) := e^{-i\theta X}$, where $X$ is a Pauli matrix \cite{nielsenchuang}. The $R_Y$ and $R_Z$ gates are similarly defined in terms of the $Y$ and $Z$ Pauli matrices. The entangling gates we use are \emph{controlled} gates of the form $CU$, where the action of a gate $U$ on a \textit{target} register depends on the state of a \textit{control} register. For example, a $CR_X$ gate applies a $R_X$ gate to the target qubit when the control qubit is in state $\ket{1}$, and performs the identity operation $I$ if the control is in state $\ket{0}$; if the qubit is in a superposition, the operation is applied to each basis element by linearity. Diagramatically, a controlled gate is represented as a vertical line connecting a box on the target register to a circle on the control register, above which we indicate the state the gate is controlled on.

Measuring a qubit in state $\ket{\psi}$ with respect to the computational basis yields an output of $0$ with probability $\norm{\braket{0}{\psi}}^2$ and output $1$ with probability $\norm{\braket{1}{\psi}}^2$. Thus, quantum computation is inherently probabilistic; in general, several runs (often called \emph{shots}) of a circuit are needed to obtain the desired result. For instance, some quantum algorithms rely on \textit{postselection}, a process which discards all executions of a quantum circuit in which a postselected measurement did not yield a desired state. Postselecting in the $\ket{0}$ state is indicated in a circuit by a $\bra{0}$ placed near the end of a wire.

An \emph{observable} is represented by a hermitian matrix $O$, i.e. a matrix satisfying $O^{\dagger} = \overline{O}^T=O$. The expectation value associated with an observable $O$ applied to a state $\ket{\psi}$ is given by
\begin{align}
    \langle O \rangle_{\ket{\psi}} := \bra{\psi} O\ket{\psi} \in \mathbb{R}
\end{align}
Examples of common observables are the Pauli matrices introduced above. For further details on quantum computation, we refer the reader to \cite{nielsenchuang}.

\section{Model} \label{sec:model}

\subsection{Unitary token embedding} \label{sec:embedding}
As a first step, our model requires a quantum representation of each element of the vocabulary. Starting with a classical vector embedding $\vec{w}$ of a token $w$, we apply a linear layer $W_E$ to obtain a set of angles $\Vec{\theta}_w = W_E \vec{w}$. These are passed to a \emph{parameterised quantum circuit} (PQC) $U$ to prepare a unitary representation of $w$, $U_w := U(\theta_{\Vec{w}})$. PQCs are a common pattern in designing quantum machine learning models, and consist of parameterised gates, the angles of which are updated as part of a training procedure \cite{benedetti2019parameterized}. 
The specific choice of circuit $U$ represents a trade-off between expressibility and circuit size, and has implications on the tractability of the gradient, an issue which we discuss in \cref{sec:limitations}.

\subsection{Mixing via Linear Combination of Unitaries} \label{sec:lcu}
Now that we have prepared unitary circuit representations of the vocabulary, we implement the mixing of token information using the \emph{linear combination of unitaries} (LCU) procedure~\cite{childs2012}. Our aim is to have this mix take the form
\begin{align}\label{eq:m}
M := \sum_{j=0}^{n-1} b_j U_j
\end{align}
for some (possibly trainable) complex parameters $b_j$  satisfying $\sum_{j=0}^{n-1} \abs{b_j} = 1$, where $n$ is the window size. For this, we require a circuit $\usel$ that applies a unitary $U_k$ out of $\{U_j\}_{j \in {0, ..., n-1}}$ to a register conditional on the state of a control register being $\ket{k}$,
\begin{align}
    &\usel = \sum_{j=0}^{n-1} \ketbra{j}{j} \otimes U_j, \qquad \usel (\ket{k} \otimes  I) = \ket{k} \otimes U_k.
\end{align}
A simple quantum circuit representation of $\usel$ is shown below.
\begin{align} \label{fig:lcu}
    U_{\text{SEL}} = \vcenter{\hbox{\begin{tikzpicture}
	\begin{pgfonlayer}{nodelayer}
		\node [style=none] (3) at (-7, 2) {};
		\node [style=none] (4) at (4, 2) {};
		\node [style=cnot ctrl] (5) at (-5, 2) {};
		\node [style=cnot ctrl] (6) at (-2, 2) {};
		\node [style=cnot ctrl] (7) at (2.25, 2) {};
		\node [style=none] (9) at (4, 0) {};
		\node [style=none] (10) at (0, 1) {\dots};
		\node [style=none] (11) at (-6.25, 0) {/};
		\node [style=none] (12) at (-6.25, 2) {/};
		\node [style=none] (13) at (-5, 2.5) {$\ket{0}$};
		\node [style=none] (14) at (-2, 2.5) {$\ket{1}$};
		\node [style=none] (15) at (2.25, 2.5) {$\ket{n-1}$};
		\node [style=none] (16) at (-5.75, 0.75) {};
		\node [style=none] (17) at (-4.25, 0.75) {};
		\node [style=none] (18) at (-5.75, -0.75) {};
		\node [style=none] (19) at (-4.25, -0.75) {};
		\node [style=none] (20) at (-5, 0) {$U_0$};
		\node [style=none] (21) at (-5, 0.75) {};
		\node [style=none] (22) at (-2.75, 0.75) {};
		\node [style=none] (23) at (-1.25, 0.75) {};
		\node [style=none] (24) at (-2.75, -0.75) {};
		\node [style=none] (25) at (-1.25, -0.75) {};
		\node [style=none] (26) at (-2, 0) {$U_1$};
		\node [style=none] (27) at (-2, 0.75) {};
		\node [style=none] (28) at (1.25, 0.75) {};
		\node [style=none] (29) at (3.25, 0.75) {};
		\node [style=none] (30) at (1.25, -0.75) {};
		\node [style=none] (31) at (3.25, -0.75) {};
		\node [style=none] (32) at (2.25, 0) {$U_{n-1}$};
		\node [style=none] (33) at (2.25, 0.75) {};
		\node [style=none] (34) at (-0.5, 2) {};
		\node [style=none] (35) at (0.5, 2) {};
		\node [style=none] (36) at (-0.5, 0) {};
		\node [style=none] (37) at (0.5, 0) {};
		\node [style=none] (8) at (-7, 0) {};
	\end{pgfonlayer}
	\begin{pgfonlayer}{edgelayer}
		\draw [style=new edge style 0] (8.center) to (36.center);
		\draw [style=new edge style 0] (37.center) to (9.center);
		\draw [style=new edge style 0] (16.center)
			 to (17.center)
			 to (19.center)
			 to (18.center)
			 to cycle;
		\draw [style=new edge style 0] (21.center) to (5);
		\draw [style=new edge style 0] (22.center)
			 to (23.center)
			 to (25.center)
			 to (24.center)
			 to cycle;
		\draw [style=new edge style 0] (6) to (27.center);
		\draw [style=new edge style 0] (31.center)
			 to (30.center)
			 to (28.center)
			 to (29.center)
			 to cycle;
		\draw [style=new edge style 0] (7) to (33.center);
		\draw [style=new edge style 0] (35.center) to (4.center);
		\draw [style=new edge style 0] (3.center) to (34.center);
	\end{pgfonlayer}
\end{tikzpicture}
}}
\end{align}
When the control register is initially in some normalised state $\ket{a} = \sum_{i=0}^{n-1} a_i \ket{i}$, the above circuit can produce the weighted superposition of unitaries $\sum_{j=0}^{n-1} |a_j|^2 U_j$. However, this is only the case if the control register is postselected to remain in state $\ket{a}$ (we refer the reader to \cite{martyn2021, dalzell2023quantum} for more details on the LCU construction). This implies that, given a unitary $\uprep$ which prepares $\ket{a}$,
\begin{align}
    \uprep\ket{0} = \ket{a} = \sum_{j=0}^{n-1} a_j \ket{j},
\end{align}
the circuit
\begin{align}
    \label{eq:lcu_circuit}
    U_M = (\uprep^\dagger \otimes I) \usel (\uprep \otimes I),
\end{align}
diagrammatically represented as
\begin{align} \label{fig:block_encoding}
    U_M = \vcenter{\hbox{\begin{tikzpicture}
	\begin{pgfonlayer}{nodelayer}
		\node [style=none] (3) at (1.75, 2) {};
		\node [style=none] (4) at (15.75, 2) {};
		\node [style=none] (8) at (1.75, 0) {};
		\node [style=none] (9) at (15.75, 0) {};
		\node [style=none] (11) at (2.5, 0) {/};
		\node [style=none] (12) at (2.5, 2) {/};
		\node [style=none] (17) at (4.75, 2) {$U_{\text{PREP}}$};
		\node [style=none] (18) at (7.75, 2.75) {};
		\node [style=none] (19) at (9.75, 2.75) {};
		\node [style=none] (20) at (7.75, -0.75) {};
		\node [style=none] (21) at (9.75, -0.75) {};
		\node [style=none] (22) at (8.75, 1) {$U_{\text{SEL}}$};
		\node [style=none] (26) at (3.5, 2.75) {};
		\node [style=none] (27) at (6, 2.75) {};
		\node [style=none] (28) at (3.5, 1.25) {};
		\node [style=none] (29) at (6, 1.25) {};
		\node [style=none] (30) at (13, 2) {$U^{\dagger}_{\text{PREP}}$};
		\node [style=none] (31) at (11.75, 2.75) {};
		\node [style=none] (32) at (14.25, 2.75) {};
		\node [style=none] (33) at (11.75, 1.25) {};
		\node [style=none] (34) at (14.25, 1.25) {};
	\end{pgfonlayer}
	\begin{pgfonlayer}{edgelayer}
		\draw [in=180, out=0] (3.center) to (4.center);
		\draw (9.center) to (8.center);
		\draw [style=new edge style 0] (18.center)
			 to (20.center)
			 to (21.center)
			 to (19.center)
			 to cycle;
		\draw [style=new edge style 0] (28.center)
			 to (26.center)
			 to (27.center)
			 to (29.center)
			 to cycle;
		\draw [style=new edge style 0] (33.center)
			 to (31.center)
			 to (32.center)
			 to (34.center)
			 to cycle;
	\end{pgfonlayer}
\end{tikzpicture}
}},
\end{align}
prepares the superposition
\begin{align} \label{eq:superposition}
   (\bra{0} \otimes I) U_M (\ket{0} \otimes I) = M = \sum_{j=0}^{n-1} |a_j|^2 U_j
\end{align}
when the control register is prepared and postselected in the $\ket{0}$ basis state. We present the full derivation of \cref{eq:superposition} in \cref{thm:block_encoding} of \cref{sec:lcu_appendix}. Conjugating $U_M$ with two projection operators yields $M$, as shown in \cref{eq:superposition}. A matrix satisfying this criterion is said to be a \textit{block encoding} of $M$. Note that the coefficients of the superposition on the right-hand side of the equation form an L1-normalised \emph{real} vector. In our model, we prepare a superposition with complex coefficients, which can be achieved by adding a phase gate to the unitary associated with each token, such that
\begin{align}
    \label{eq:complex_lcu_coefficients}
    M:= \sum_{j=0}^{n-1} |a_j|^2 U'_j = \sum_{j=0}^{n-1} e^{i\gamma_j}|a_j|^2 U_j.
\end{align}
Throughout the rest of the text, we assume the coefficients of the LCU to be $\{b_j\}_{j \in 0 \dots n-1}$, where $b_j := e^{i\gamma_j}|a_j|^2 \in \mathbb{C}$.

\subsection{Nonlinearity via Quantum Singular Value Transformation}
\label{sec:qsvt}
The tools discussed thus far provide a method to prepare a linear combination of token embeddings encoded as unitary matrices. To provide richer interactions between the token unitaries, we use a method to prepare nonlinear transformations of this superposition. For this, we employ the Quantum Singular Value Transform (QSVT)~\cite{Gilyn2019}, which is able to apply polynomial transformations to a block-encoded matrix.
Given a block encoding $U_M$ of a matrix $M$, and a real polynomial $P_{\Vec{c}}$ of degree $d$, defined as:
\begin{align}
    P_{\Vec{c}}(x) &= c_d \cdot x^d +  c_{d-1} \cdot x^{d-1} + \dots + c_1 \cdot x + c_0
\end{align}
with the condition that $P_{\Vec{c}}$ obeys
\begin{align}
    |P_{\Vec{c}}(x)| \leq 1&,\qquad \forall x \in [-1, 1]\\
    \textrm{parity}(P_{\Vec{c}}) = d\;\textrm{mod}\;2&
\end{align}
the QSVT provides a circuit which prepares the matrix:
\begin{align}
\label{eq:polynomial}
    P_{\Vec{c}}(M) &= c_d M^d +  c_{d-1} M^{d-1} + \dots + c_1 M + c_0 I
\end{align}
The full expression for the QSVT is
\begin{align}
\begin{cases}\Pi_{\phi_1}U_M\left[\prod_{k=1}^{\frac{d-1}{2}}\Pi_{\phi_{2k}}U_M^\dagger\Pi_{\phi_{2k+1}}U_M\right], \qquad &d \; \text{odd}, \\
\left[\prod_{k=1}^{\frac{d}{2}}\Pi_{\phi_{2k-1}}U_M^\dagger\Pi_{\phi_{2k}}U_M\right], \qquad &d \; \text{even},
\end{cases}
\end{align}
where $\Pi_\phi := e^{i\phi(2\ketbra{0}{0}-I)}$ acts as a controlled $R_Z$ rotation on a newly introduced ancilla qubit, and the angles $\phi_k$ are fixed and chosen so as to implement the desired polynomial (see~\cite{martyn2021, Gilyn2019}). A circuit implementing the QSVT for a cubic polynomial (i.e. $d=3$) is illustrated below, where $\oplus$ is a NOT operation that is applied if the control register is in state $\ket{0}$; this condition is represented by the projection $\Pi := \ketbra{0}{0}$. The blue box highlights the subcircuit which is repeated for higher degree polynomials.
\begin{align} \label{fig:qsvt}
    \vcenter{\hbox{\scalebox{0.8}{\begin{tikzpicture}
	\begin{pgfonlayer}{nodelayer}
		\node [style=none] (0) at (-6.25, 2) {};
		\node [style=none] (1) at (24, 2) {};
		\node [style=none] (2) at (-6.25, 0) {};
		\node [style=none] (3) at (24, 0) {};
		\node [style=none] (4) at (-5.5, 0) {/};
		\node [style=none] (5) at (-5.5, 2) {/};
		\node [style=none] (8) at (5.5, 2.75) {};
		\node [style=none] (9) at (7.5, 2.75) {};
		\node [style=none] (10) at (5.5, -0.75) {};
		\node [style=none] (11) at (7.5, -0.75) {};
		\node [style=none] (12) at (6.5, 1) {$U_M^{\dagger}$};
		\node [style=none] (13) at (-7.25, 2) {$\ket{0}$};
		\node [style=none] (14) at (-7.25, 0) {$\ket{0}$};
		\node [style=none] (15) at (25.25, 2) {$\bra{0}$};
		\node [style=none] (16) at (-6.25, 5) {};
		\node [style=none] (17) at (24, 5) {};
		\node [style=cnot targ] (19) at (0, 5) {};
		\node [style=cnot targ] (22) at (4, 5) {};
		\node [style=none] (23) at (-7.25, 5) {$\ket{0}$};
		\node [style=none] (34) at (-0.5, 2.75) {};
		\node [style=none] (35) at (0.5, 2.75) {};
		\node [style=none] (36) at (-0.5, 1.25) {};
		\node [style=none] (37) at (0.5, 1.25) {};
		\node [style=none] (38) at (0, 2) {$\Pi$};
		\node [style=none] (39) at (0, 2.75) {};
		\node [style=none] (40) at (3.5, 2.75) {};
		\node [style=none] (41) at (4.5, 2.75) {};
		\node [style=none] (42) at (3.5, 1.25) {};
		\node [style=none] (43) at (4.5, 1.25) {};
		\node [style=none] (44) at (4, 2) {$\Pi$};
		\node [style=none] (45) at (4, 2.75) {};
		\node [style=none] (61) at (-3.5, 2.75) {};
		\node [style=none] (62) at (-1.5, 2.75) {};
		\node [style=none] (63) at (-3.5, -0.75) {};
		\node [style=none] (64) at (-1.5, -0.75) {};
		\node [style=none] (65) at (-2.5, 1) {$U_M$};
		\node [style=none] (66) at (0.75, 5.75) {};
		\node [style=none] (67) at (3.25, 5.75) {};
		\node [style=none] (68) at (0.75, 4.25) {};
		\node [style=none] (69) at (3.25, 4.25) {};
		\node [style=none] (70) at (2.025, 5) {$R_z(\phi_2)$};
		\node [style=none] (71) at (2, 5.75) {};
		\node [style=cnot targ] (72) at (9, 5) {};
		\node [style=cnot targ] (73) at (13, 5) {};
		\node [style=none] (74) at (8.5, 2.75) {};
		\node [style=none] (75) at (9.5, 2.75) {};
		\node [style=none] (76) at (8.5, 1.25) {};
		\node [style=none] (77) at (9.5, 1.25) {};
		\node [style=none] (78) at (9, 2) {$\Pi$};
		\node [style=none] (79) at (9, 2.75) {};
		\node [style=none] (80) at (12.5, 2.75) {};
		\node [style=none] (81) at (13.5, 2.75) {};
		\node [style=none] (82) at (12.5, 1.25) {};
		\node [style=none] (83) at (13.5, 1.25) {};
		\node [style=none] (84) at (13, 2) {$\Pi$};
		\node [style=none] (85) at (13, 2.75) {};
		\node [style=none] (86) at (9.75, 5.75) {};
		\node [style=none] (87) at (12.25, 5.75) {};
		\node [style=none] (88) at (9.75, 4.25) {};
		\node [style=none] (89) at (12.25, 4.25) {};
		\node [style=none] (90) at (11.025, 5) {$R_z(\phi_1)$};
		\node [style=none] (91) at (11, 5.75) {};
		\node [style=cnot targ] (92) at (18, 5) {};
		\node [style=cnot targ] (93) at (22, 5) {};
		\node [style=none] (94) at (17.5, 2.75) {};
		\node [style=none] (95) at (18.5, 2.75) {};
		\node [style=none] (96) at (17.5, 1.25) {};
		\node [style=none] (97) at (18.5, 1.25) {};
		\node [style=none] (98) at (18, 2) {$\Pi$};
		\node [style=none] (99) at (18, 2.75) {};
		\node [style=none] (100) at (21.5, 2.75) {};
		\node [style=none] (101) at (22.5, 2.75) {};
		\node [style=none] (102) at (21.5, 1.25) {};
		\node [style=none] (103) at (22.5, 1.25) {};
		\node [style=none] (104) at (22, 2) {$\Pi$};
		\node [style=none] (105) at (22, 2.75) {};
		\node [style=none] (106) at (14.5, 2.75) {};
		\node [style=none] (107) at (16.5, 2.75) {};
		\node [style=none] (108) at (14.5, -0.75) {};
		\node [style=none] (109) at (16.5, -0.75) {};
		\node [style=none] (110) at (15.5, 1) {$U_M$};
		\node [style=none] (111) at (18.75, 5.75) {};
		\node [style=none] (112) at (21.25, 5.75) {};
		\node [style=none] (113) at (18.75, 4.25) {};
		\node [style=none] (114) at (21.25, 4.25) {};
		\node [style=none] (115) at (20.025, 5) {$R_z(\phi_0)$};
		\node [style=none] (116) at (20, 5.75) {};
		\node [style=none] (117) at (-4, 6) {};
		\node [style=none] (118) at (-4, -1.5) {};
		\node [style=none] (119) at (14, 6) {};
		\node [style=none] (120) at (14, -1.5) {};
		\node [style=none] (121) at (24, 5) {};
		\node [style=none] (122) at (24, 5.375) {};
		\node [style=none] (123) at (24, 4.625) {};
		\node [style=none] (124) at (24.1, 5.25) {};
		\node [style=none] (125) at (24.1, 4.75) {};
		\node [style=none] (126) at (24.2, 5.175) {};
		\node [style=none] (127) at (24.2, 4.825) {};
		\node [style=none] (128) at (24, 5) {};
	\end{pgfonlayer}
	\begin{pgfonlayer}{edgelayer}
		\draw (0.center) to (1.center);
		\draw (3.center) to (2.center);
		\draw [style=new edge style 0] (8.center)
			 to (10.center)
			 to (11.center)
			 to (9.center)
			 to cycle;
		\draw (16.center) to (17.center);
		\draw [style=new edge style 0] (34.center)
			 to (35.center)
			 to (37.center)
			 to (36.center)
			 to cycle;
		\draw (39.center) to (19);
		\draw [style=new edge style 0] (41.center)
			 to (43.center)
			 to (42.center)
			 to (40.center)
			 to cycle;
		\draw (45.center) to (22);
		\draw [style=new edge style 0] (61.center)
			 to (63.center)
			 to (64.center)
			 to (62.center)
			 to cycle;
		\draw [style=new edge style 0] (67.center)
			 to (69.center)
			 to (68.center)
			 to (66.center)
			 to cycle;
		\draw [style=new edge style 0] (74.center)
			 to (75.center)
			 to (77.center)
			 to (76.center)
			 to cycle;
		\draw (79.center) to (72);
		\draw [style=new edge style 0] (81.center)
			 to (83.center)
			 to (82.center)
			 to (80.center)
			 to cycle;
		\draw (85.center) to (73);
		\draw [style=new edge style 0] (87.center)
			 to (89.center)
			 to (88.center)
			 to (86.center)
			 to cycle;
		\draw [style=new edge style 0] (94.center)
			 to (95.center)
			 to (97.center)
			 to (96.center)
			 to cycle;
		\draw (99.center) to (92);
		\draw [style=new edge style 0] (101.center)
			 to (103.center)
			 to (102.center)
			 to (100.center)
			 to cycle;
		\draw (105.center) to (93);
		\draw [style=new edge style 0] (106.center)
			 to (108.center)
			 to (109.center)
			 to (107.center)
			 to cycle;
		\draw [style=new edge style 0] (112.center)
			 to (114.center)
			 to (113.center)
			 to (111.center)
			 to cycle;
		\draw [style=blue edge] (117.center) to (118.center);
		\draw [style=blue edge] (119.center) to (117.center);
		\draw [style=blue edge] (118.center) to (120.center);
		\draw [style=blue edge] (120.center) to (119.center);
		\draw (122.center) to (123.center);
		\draw (124.center) to (125.center);
		\draw (126.center) to (127.center);
	\end{pgfonlayer}
\end{tikzpicture}
}}}
\end{align}

We would like to be able to effect polynomial transformations of $M$ which do not have parity $\in \{0, 1\}$. For this, we observe that any polynomial can be split into the sum of an even (all even powers) and an odd (all odd powers) parity polynomial
\begin{align}
    P_{\Vec{c}} = P_{\textit{odd}} + P_{\textit{even}}.
\end{align}
Using this decomposition, it is possible to prepare $P_{\Vec{c}}(M)$ with unknown parity using the following LCU circuit, using one additional postselected ancilla qubit.

\begin{align} \label{fig:arbitrary_poly}
    \vcenter{\hbox{\scalebox{0.8}{\begin{tikzpicture}
	\begin{pgfonlayer}{nodelayer}
		\node [style=none] (0) at (19.5, 1.25) {};
		\node [style=none] (2) at (19.5, 0) {};
		\node [style=none] (3) at (32, 0) {};
		\node [style=none] (4) at (20.25, 0) {/};
		\node [style=none] (5) at (20.25, 1.25) {/};
		\node [style=none] (11) at (18.5, 1.25) {$\ket{0}$};
		\node [style=none] (12) at (18.5, 0) {$\ket{0}$};
		\node [style=none] (14) at (19.5, 5) {};
		\node [style=none] (15) at (32, 5) {};
		\node [style=none] (18) at (18.5, 5) {$\ket{0}$};
		\node [style=none] (76) at (22.25, 2.75) {};
		\node [style=none] (77) at (24.25, 2.75) {};
		\node [style=none] (78) at (22.25, -0.75) {};
		\node [style=none] (79) at (24.25, -0.75) {};
		\node [style=none] (80) at (23.25, 1) {$P_{\textit{odd}}$};
		\node [style=none] (87) at (23.25, 2.75) {};
		\node [style=cnot ctrl] (88) at (23.25, 5) {};
		\node [style=none] (89) at (19.5, 2.25) {};
		\node [style=none] (91) at (18.5, 2.25) {$\ket{0}$};
		\node [style=none] (92) at (26, 2.75) {};
		\node [style=none] (93) at (28, 2.75) {};
		\node [style=none] (94) at (26, -0.75) {};
		\node [style=none] (95) at (28, -0.75) {};
		\node [style=none] (96) at (27, 1) {$P_{\textit{even}}$};
		\node [style=none] (97) at (27, 2.75) {};
		\node [style=vertex set] (98) at (27, 5) {};
		\node [style=none] (100) at (32, 1.25) {};
		\node [style=none] (105) at (32, 2.25) {};
		\node [style=none] (106) at (32.75, 1.25) {$\bra{0}$};
		\node [style=none] (107) at (32.75, 5) {$\bra{0}$};
		\node [style=none] (108) at (28.75, 5.75) {};
		\node [style=none] (109) at (30.25, 5.75) {};
		\node [style=none] (110) at (28.75, 4.25) {};
		\node [style=none] (111) at (30.25, 4.25) {};
		\node [style=none] (112) at (29.525, 5) {$H$};
		\node [style=none] (113) at (20.5, 5.75) {};
		\node [style=none] (114) at (22, 5.75) {};
		\node [style=none] (115) at (20.5, 4.25) {};
		\node [style=none] (116) at (22, 4.25) {};
		\node [style=none] (117) at (21.275, 5) {$H$};
		\node [style=none] (118) at (32, 2.625) {};
		\node [style=none] (119) at (32, 1.875) {};
		\node [style=none] (120) at (32.1, 2.5) {};
		\node [style=none] (121) at (32.1, 2) {};
		\node [style=none] (122) at (32.2, 2.425) {};
		\node [style=none] (123) at (32.2, 2.075) {};
		\node [style=none] (124) at (32, 2.25) {};
	\end{pgfonlayer}
	\begin{pgfonlayer}{edgelayer}
		\draw (0.center) to (100.center);
		\draw (89.center) to (105.center);
		\draw (3.center) to (2.center);
		\draw (14.center) to (15.center);
		\draw [style=new edge style 0] (76.center)
			 to (78.center)
			 to (79.center)
			 to (77.center)
			 to cycle;
		\draw (88) to (87.center);
		\draw [style=new edge style 0] (95.center)
			 to (93.center)
			 to (92.center)
			 to (94.center)
			 to cycle;
		\draw (98) to (97.center);
		\draw [style=new edge style 0] (109.center)
			 to (111.center)
			 to (110.center)
			 to (108.center)
			 to cycle;
		\draw [style=new edge style 0] (115.center)
			 to (113.center)
			 to (114.center)
			 to (116.center)
			 to cycle;
		\draw (118.center) to (119.center);
		\draw (120.center) to (121.center);
		\draw (122.center) to (123.center);
	\end{pgfonlayer}
\end{tikzpicture}
}}}
\end{align}
\subsection{Quixer} \label{sec:quixer}
Having described each of the components of Quixer, we now provide an end-to-end description of the model. Given an input sequence $\{w_j\}_{j\in 0 \dots n-1}$, Quixer begins by preparing a $q$-qubit unitary circuit $U_{w_j} := U(\theta_{\Vec{w}_j})$ for each token, using a trainable linear matrix $W_E$ to obtain a set of PQC for each token embedding $\theta_{\Vec{w}} = W_E \Vec{w}$, as described in \cref{sec:embedding}. 

An LCU circuit $U_M$ is then instantiated with trainable complex coefficients $\{b_j\}_{j\in 0 \dots n-1}$, to prepare a linear combination of the token unitaries $M_{\vec{b}, \theta} = \sum_j^{n-1} b_j U_{w_j}$ as described in \cref{sec:lcu}. Next, a polynomial $P_{\Vec{c}}$ of degree $d$ with trainable coefficients $\Vec{c}$ is applied to the LCU via a QSVT to obtain $P_{\Vec{c}}(M_\theta)$, as outlined in \cref{sec:qsvt}. 
This circuit is then applied to a $\ket{0}$ state on the data register, followed by a trainable PQC $U_{\textit{FF}}$ on the data register, resulting in the final (unnormalised) quantum state
\begin{align}
    \ket{\psi} = U_{\textit{FF}}P_{\Vec{c}}\big(M_{\vec{b}, \theta}\big)\ket{0}.
\end{align}
Information from this state is then read out by measuring multiple expectation values
\begin{align}
    o_k =  \langle O_k \rangle_{\ket{\psi}},
\end{align}
resulting in a classical real vector $\Vec{o}$. The final output of the model is then obtained by applying a fully-connected feed-forward neural network $f_{\textit{out}}$ to $\vec{o}$,
\begin{align}
    \Vec{y} = f_{\textit{out}}(\vec{o}).
\end{align}

\subsubsection{Attention in Quixer}
\label{sec:attention}
In the classical transformer, the dot product self-attention mechanism operates by preparing a sum of \textit{value} vectors, weighted by pairwise interactions captured using \textit{query}-\textit{key} products \cite{vaswani2017}. In contrast, Quixer captures interactions between multiple tokens through the composition of their representative unitaries.  A weighted sum of such interactions is computed to prepare the final state. While sequential composition with a unitary preserves the norm of the quantum state, it changes the magnitude of the expectation value for a particular observable. This allows Quixer to model attention through unitary composition. For instance, when implementing a quadratic polynomial using the QSVT procedure, Quixer prepares the quantum state described by
\begin{align}
    P_{\vec{c}}(M_{\Vec{b}, \theta}) \ket{0} &= c_2 \sum_{j,k = 0}^{n-1} b_j b_k U_j U_k \ket{0} + c_1 \sum_{j=0}^{n-1} b_j U_j \ket{0} + c_0 \ket{0}.
\end{align}
The action of this model can be seen as the sum of pairwise interactions (again, captured through composition of word unitaries), single-token terms, and a final bias term. Here pairwise interactions are computed between all \textit{skip-bigrams} (pairs of tokens, not necessarily adjacent) in the context. Higher degree polynomials, analogously, compute interactions between skip-$k$-grams; we write this out in \cref{sec:attention_higher_d}.

\subsection{Resource estimates}
\label{sec:resources}

\subsubsection{Postselection probability} \label{sec:postselection}
As discussed in \cref{sec:lcu}, the preparation of the LCU $M_{\vec{b}, \theta}$ is contingent on postselecting the control register in state $\ket{0}$. The probability of this succeeding, and of producing the desired state $M_{\vec{b}, \theta}\ket{0}$, is equal to
\begin{align}
    p_{M} &= \norm{M_{\vec{b}, \theta}\ket{0}}^2 \\
    &=\bra{0} M_{\vec{b}, \theta}^\dagger M_{\vec{b}, \theta}\ket{0} \\
    &=\sum_{j,k=0}^{n-1} \overline{b}_j b_k \bra{0}U_j^\dagger U_k \ket{0} \\
    &= \sum_{j < k} \overline{b}_j b_k \bra{0}U_j^\dagger U_k\ket{0} + \overline{b}_k b_j\bra{0}U_k^\dagger U_j\ket{0} + \sum_{j=0}^{n-1} |b_j|^2 \\
    &= \sum_{j < k} 2 \cdot \mathrm{Re}[\overline{b}_j b_k\bra{0}U_j^\dagger U_k\ket{0}] + \sum_{j=0}^{n-1} |b_j|^2.
\end{align}
In turn, the probability of successfully preparing the desired polynomial using the QSVT is
\begin{align} \label{eq:final_postselection}
    p = \norm{P(M_{\Vec{b}, \theta})\ket{0}}^2.
\end{align}

However, neither expression yields a lower bound on the success probability, since the minimum value of both is 0 without further assumptions on the circuit structure. Previous works have lower bounded the success probability of LCU and QSVT circuits through spectral analysis of the matrices $M$ and $P(M)$~\cite{Gilyn2019, watts2023quantum}. Such analysis, however, requires additional assumptions on the training algorithm, parameter space and unitary token embedding choice, as it depends on the concrete form that these unitaries take. Thus, in general, our model is not immediately amenable to such analysis without implementing further constraints. For the concrete Quixer implementation provided in \cref{sec:results}, we evaluate the postselection probability empirically.

\subsubsection{Gate and qubit complexity}
\label{sec:complexity}
Given a sequence of length $n$, our model uses $q$ qubits for the data register, $\lceil \log_2(n) \rceil$ qubits for the control register, and 3 ancillae: one for the QSVT procedure, one for the combination of odd and even parity terms, and one for the implementation of the QSVT projectors $\Pi$ (detailed below, and in \cref{sec:toffoli}). Thus, the asymptotic total number of qubits is
\begin{align}
\label{eq:qubit_complexity}
    O(q + \mathrm{log}_2(n)).
\end{align}

The runtime of a quantum circuit is characterised by its gate complexity. Following the construction in \cref{sec:toffoli}, a single-qubit gate can be controlled on a given number of qubits using a gate count linear in that number. Let us assume that each token unitary has at most $g$ gates when written in terms of single-qubit operations controlled on at most one qubit (note that $g$ is well-defined, since any unitary can be decomposed into $CX$ gates and single-qubit gates~\cite{nielsenchuang}). Then, each token unitary controlled on $\log_2(n)$ qubits can be implemented using $O(g\log_2(n))$ gates, implying that the LCU component of the circuit can be implemented using $O(n g \log_2(n))$ gates. As explained in \cref{sec:toffoli}, each projection $\Pi$ in the QSVT can also be implemented using a number of gates linear in the number of qubits in the control register. For a QSVT implementing a polynomial of degree $d$, this yields an overall asymptotic gate count of $O(d n  g \cdot \log_2(n))$
If each unitary is a PQC with $l$ layers, each of which contains a number of parameterised gates proportional to the qubits in the data register (i.e. $g \propto ql$), this corresponds to $O(d n  q l \cdot \log_2(n))$.

The above complexity can be further improved by leveraging the technique in \cite{Babbush2018}. For an additional $\log_2(n)-2$ ancilla qubits (which does not change the qubit complexity in \cref{eq:qubit_complexity}), each token unitary controlled on $\log_2(n)$ qubits can be implemented using $O(g)$ instead of  $O(g\log_2(n))$ gates. This results in a gate complexity of 
\begin{align}
    O(d n g),
\end{align}
or, again, if each unitary is a PQC with $l$ layers, each of which contains a number of parameterised gates proportional to the qubits in the data register (i.e. $g \propto ql$),
\begin{align}
    O(d n  q l).
\end{align}

\section{Experimental results} \label{sec:results}

\subsection{Setup}\label{sec:setup}
To evaluate our model in a practical setting, we apply an instance of Quixer to a language modelling task. Here, given a sequence of words $\{w_j\}_{j \in 0\dots n-1}$ (which we take to be our tokens), the model must predict the subsequent word $w_n$. We evaluate our model on the Penn Treebank (PTB) dataset, which consists of 966K training tokens, 77K validation tokens, and 86K test tokens \cite{marcus1993building}. We obtain this from the HuggingFace datasets package\footnote{\url{https://huggingface.co/datasets/ptb_text_only}}.

We implemented Quixer as a Torch module, and used TorchQuantum \cite{hanruiwang2022quantumnas}, a Torch-native quantum computation framework, to simulate the PQCs and compute expectation values. When simulating the model classically, it is not necessary to prepare the explicit circuits for the LCU and QSVT operations. Instead, we apply each token unitary to a copy of the data register, and directly prepare a linear combination weighted by $\vec{b}$. The polynomial implemented by the QSVT can then be directly computed by repeating the LCU application on the data register $d$ times. For further details on the practical classical simulation of Quixer, we refer the reader to the code provided in the GitHub repository \href{https://github.com/CQCL/Quixer/}{github.com/CQCL/Quixer}.

Each token unitary is composed of 4 layers of ``circuit 14'' out of several parameterised circuits that \citet{Sim2019} studies. Our choice is motivated by the high expressibility and entangling capability of this circuit. Note, however, that this choice is not canonical, and any parameterised circuit can be employed in our architecture. One layer of this circuit alternates layers of $R_Y$ and $CR_X$ gates as follows
\begin{align} \label{eq:a14}
    \prod_{j=1}^q CR_{X_{j,j+1}}(\theta_{4,j}) \prod_{j=1}^q R_{Y_j}(\theta_{3,j}) \prod_{j=1}^q CR_{X_{j,j-1}} (\theta_{2,j}) \prod_{j=1}^q R_Y(\theta_{1,j}).
\end{align}
For $l$ layers of circuit 14 acting on $q$ qubits, the number of parameterised gates (and parameters) is $4lq$. Below, we show a single layer of this circuit on $3$ qubits. 
\begin{align}
\label{fig:a14}
\begin{adjustbox}{width=0.75\textwidth}
\begin{quantikz}
\lstick{} & \gate{R_Y} & \ctrl{1} & \qw & \gate{R_X} & \gate{R_Y} & \gate{R_X} & \qw & \ctrl{2} & \qw\\
\lstick{} & \gate{R_Y} & \gate{R_X} & \ctrl{1} & \qw & \gate{R_Y} & \ctrl{-1} & \gate{R_X} & \qw & \qw \\
\lstick{} & \gate{R_Y} & \qw & \gate{R_X} & \ctrl{-2} & \gate{R_Y} & \qw & \ctrl{-1} & \gate{R_X} & \qw 
\end{quantikz}
\end{adjustbox}
\end{align}
We implement a trainable cubic polynomial. For this configuration, each token unitary has 96 parameters, which we prepare by applying $W_E$ to a 512-dimensional word embedding.  In our experiments, we compute the expectation values of the $X$, $Y$ and $Z$ Pauli operators independently for each qubit, resulting in a vector $\Vec{o} \in \mathbb{R}^{3q}$. $f_{out}$, consisting of a 2-layer feed-forward neural network with a ReLu nonlinearity in between, is used to map the expectation values to a probability distribution over tokens.

We compare Quixer against an LSTM \cite{hochreiter1997long} and the Transformer~\cite{vaswani2017}, as provided in the PyTorch~\cite{paszke2019pytorch} package, and a PyTorch implementation of the FNet~\cite{lee2022fnet}, which we include in this project's source code. The LSTM represents a high-performant, non-attention-based architecture. FNet is a simplified version of the Transformer architecture which replaces the multi-head self-attention unit with a 2-dimensional Fourier transform applied to the input matrix. The residuals, layer normalisation and MLPs from the original Transformer model are retained. Finally, the Transformer is the main component of the majority of modern language models.

Quixer processes a sequence of tokens of length $n$ and produces a single subsequent token as output. We use a context length of $32$, and stride by $1$ token per step to generate each token in the dataset. We provide the classical baselines with identical setup of window size and stride. The performance of each model is measured by \textit{perplexity} (PPL), which is defined as the exponent of the cross entropy loss; see \cite{jurafsky2022speech} for more details. For each of the classical baselines, we use embedding sizes of $96$ and $128$. An embedding dimension of $96$ restricts the classical models to the number of angles Quixer uses to parameterise each token unitary.

All models were trained using the Adam optimiser \cite{kingma2014adam}, and the learning rate was varied according to a cosine annealing schedule \cite{loshchilov2016sgdr}. Each batch contains $32$ contexts, and each context produces $32$ tokens, yielding an effective batch size of $1024$. All models were trained for $30$ epochs, and the best epoch was selected based on perplexity on the validation set. Learning rates, dropout, and weight decay were tuned per model, and are described in \cref{sec:hyperparameters}. Quixer was trained on one A100 GPU, on which $30$ epochs took $3$ hours $45$ minutes.

\subsection{Results}

\Cref{tab:ppltab} shows the perplexities obtained by each model for the word-level language modelling task on the PTB dataset. All results are reported averaged over $10$ runs, along with error bars of one standard deviation. We observe that Quixer outperforms both sizes of LSTM, performs competitively with the $96$-dimension FNet, and only marginally worse than the $128$-dimension FNet. The Transformer outperforms all other models, even with a $96$-dimension embedding.

We obtain results similar to FNet, a model which is known to achieve results competitive with Transformer-based language models at scale \cite{lee2022fnet}. This is an encouraging result for a first quantum transformer applied to language modelling, and provides a validation of the Quixer architecture. Note, however, that these results are far from those obtained by state-of-the-art models, and that a direct comparison with such models is not the objective of this work.

\begin{table}[ht]
\centering
\def\arraystretch{1.25}%
\setlength{\tabcolsep}{1.25em}
\begin{tabular}{cccc}
\hline
\textbf{Model}               & \textbf{Dimension} & \textbf{Layers} & \textbf{PPL}   \\ \hline\hline
\multirow{2}{*}{LSTM}        & 96                 & 2               & 144.3 (\textpm 9.1) \\
                             & 128                & 2               & 127.1 (\textpm 3.1) \\ \hline
\multirow{2}{*}{FNet}        & 96                 & 2               & 120.5 (\textpm 1.0) \\
                             & 128                & 2               & 117.7 (\textpm 0.8) \\ \hline
\multirow{2}{*}{Transformer} & 96                 & 1               & 100.1 (\textpm 0.2) \\
                             & 128                & 1               & \textbf{97.0 (\textpm 0.3)}  \\ \hline\hline
Quixer (Ours)                       & 6 qubits           & cubic           & 122.0 (\textpm 2.2) \\ \hline
\end{tabular}
\vspace{0.75em}
\caption{Results for word-level language modelling on Penn Treebank}
\label{tab:ppltab}
\end{table}

\Cref{fig:success_violin} plots the distribution of postselection probabilities on the PTB test dataset for $10$ runs of our model with the same hyperparameters, but different seeds. The minimum of mean success probabilities across a run is 0.0159. The mean success probability across all runs is 0.0757~($\pm$0.0460). As a reference, the Boltzmann machine, a model which can be recovered as a special case of Quixer (see \cref{sec:framework}), involves the preparation of Gibbs states~\cite{Gilyn2019, coopmans2023sample}. This has a success probability lower bound given by $2^{-q}$, which evaluates to $0.0156$ for $q=6$. Further analysis is necessary to characterise the scaling of our model's success probability with system size.

\begin{figure}[ht]
    \centering
    \includegraphics[scale=0.6]{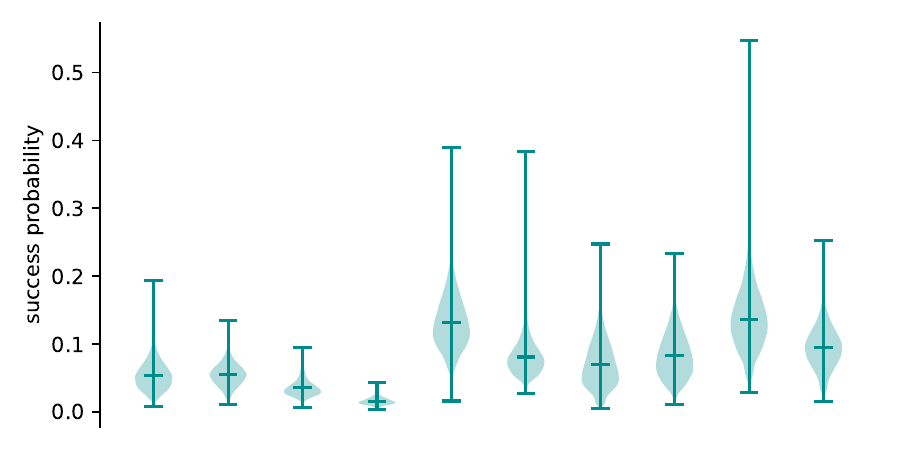}
    \caption{Distribution of postselection success probabilities (as in \cref{eq:final_postselection}) across 10 different seeds. Horizontal bars indicate mean and extrema for each seed.}
    \label{fig:success_violin}
\end{figure}

\section{Quixer as a framework}
\label{sec:framework}
Our experiments employ a very flexible instance of Quixer, with several sets of parameters trained in-task. Subsets of these, however, may be fixed \textit{a priori} to yield new models. The polynomial coefficients $\vec{c}$ need not be trained in-task, and specific polynomials may be chosen to achieve a desirable balance between expressibility, success probabilities, and magnitudes of the gradients. For example, fixing the polynomial to be an approximation of the matrix exponential and the token unitaries to be single Pauli operators, one recovers the quantum Boltzmann machine~\cite{coopmans2023sample}.

The instance of Quixer presented in \cref{sec:results} is a single layer implementing a cubic nonlinearity. By repeating the QSVT circuit for a specific polynomial $P_{\Vec{c}}$, with different LCU encodings of the data, it is possible to meaningfully extend Quixer to a multi-layer setting. For instance, a 2-layer Quixer model implemented in this manner would prepare the state
\begin{align}
    \ket{\psi} = U_{\textit{FF}}P_{\Vec{c}}\big(M_{\vec{b'}, \theta'}\big) P_{\Vec{c}}\big(M_{\vec{b}, \theta}\big)\ket{0}.
\end{align}

Note also that the token unitaries $U_w$ in our experiments were chosen to be highly-expressive PQCs. This is known to hamper the trainability of the model as the number of qubits increases~\cite{McClean2018}. This expressivity may be adjusted to mitigate this problem, for example through the use of matchgate circuits \cite{matos2023}. We discuss this issue further in  \cref{sec:limitations}.

\section{Limitations} \label{sec:limitations}
A significant challenge faced by contemporary quantum machine learning models is that the currently available methods to obtain gradients on a quantum computer have been shown to take time polynomial in the number of parameters~\cite{abbas2023on}, which is prohibitive if this number is to approach those used in modern large language models. Another well-known problem faced by quantum machine learning models is a concentration of measure phenomenon that causes gradients to exponentially vanish as the number of qubits in the model increases~\cite{McClean2018}. While some quantum models are not affected by this, such as those comprised of matchgate circuits~\cite{matos2023}, it is believed that most models evading this issue can be simulated classically~\cite{cerezo2024does}, precluding any quantum advantage. Finding an instance of Quixer that does not suffer from vanishing gradients while being expressive enough to not be amenable to classical simulation is left to future work. 

As outlined in \cref{sec:results}, training and evaluation has been done fully classically in this work. While this is representative of the model's performance, it is necessary to consider implementation overheads on a real device, along with any noise if running on an architecture which is not fully fault tolerant. Moreover, comparisons between classical models and quantum counterparts can be difficult due to their inherently different data representations. While typical classical models manipulate parameters and intermediate representations as real vectors, quantum models are limited to low-dimensional parameterisation of unitary matrices on exponential-sized complex-valued systems. Another limitation of the results presented here is the relatively small scale of models considered, as the exponential scaling of vector space dimension with the number of qubits precludes the simulation of large-scale quantum systems.

\section{Conclusion} \label{sec:conclusion}
In this work, we have described Quixer, a new quantum transformer architecture. We successfully applied it to a language modelling task on the Penn Treebank dataset, obtaining encouraging results. This represents the first example of a quantum transformer applied to a real-world language modelling task. We further described how Quixer can be extended to a family of quantum transformer models by e.g. making a particular choice of the coefficients of the polynomial transformation effected by the QSVT. Future work will involve finding new instances of this framework which make favourable trade-offs between gradient magnitudes and classical simulability, to scale the model up closer to the performance of contemporary classical language models.

\section{Acknowledgements}

We thank Tuomas Laakkonen, Fr\'{e}d\'{e}ric Sauvage, Marcello Benedetti and Konstantinos Meichanetzidis for feedback on this manuscript. We further thank Tuomas Laakkonen for insightful discussions, for suggesting the use of the Toffoli construction in~\cite{ancillas}, and for suggesting the use of the technique in~\cite{Babbush2018} to reduce the complexity of our model in \cref{sec:complexity}.

{
\small

\bibliography{bibliography.bib}
}

\appendix

\section{Appendix}

\subsection{Attention in Quixer, $d > 2$}
\label{sec:attention_higher_d}

As mentioned in \cref{sec:attention}, higher degree polynomials in the QSVT compute interactions between skip-k-grams. Indeed, from \cref{eq:m} and \cref{eq:polynomial}, 
\begin{align}
    P(M_{\Vec{b}, \theta}) &= P\left (\sum_j^{n-1} b_j U_j \right) \\
    &= \sum_{k=0}^d c_k \left(\sum_{j=0}^{n-1} b_j U_j\right)^k \\
    &= \sum_{k=1}^d c_k \left [ \sum_{\vec{\alpha} \in \{1,...,n-1\}^k} b_{\alpha_1} ... b_{\alpha_k} U_{\alpha_1} ... U_{\alpha_k} \right ] + c_0 I.
\end{align}
This is the sum of the composition of unitaries associated with all skip-$k$-grams, for $k \in \{1 \dots d\}$.

\subsection{Auxiliary result}
\label{sec:lcu_appendix}

The following lemma proves that the LCU circuit defined in \cref{eq:lcu_circuit} yields the desired result.
\begin{lemma}
\label{thm:block_encoding}
    For the circuit defined in \cref{eq:lcu_circuit}, it is the case that
    \begin{align}
       (\bra{0} \otimes I) U_M (\ket{0} \otimes I) = \sum_{i=0}^{n-1} |a_i|^2 U_i.
    \end{align}
\end{lemma}
\begin{proof}
\begin{align}
    M :=& (\bra{0} \otimes I) U_M (\ket{0} \otimes I)\\
    =& (\bra{0}\uprep^{\dagger}  \otimes I)\usel (\uprep\ket{0} \otimes I)\\
    =& (\bra{a} \otimes I) \usel (\ket{a} \otimes I)\\
    =& \big(\sum_{i=0}^{n-1} \overline{a}_i \bra{i} \otimes I\big) \sum_{j=0}^{n-1} \ketbra{j}{j} \otimes U_j \big(\sum_{k=0}^{n-1} a_k \ket{k} \otimes I\big)\\
    =& \sum_{i,j,k=0}^{n-1} \delta_{ij} \delta_{jk} a_i \overline{a}_k U_j\\
    =& \sum_i^{n-1} |a_i|^2 U_i
\end{align}
\end{proof}

\subsection{Toffoli construction} \label{sec:toffoli}
\citet{ancillas} describes a method to construct a Toffoli gate~\cite{nielsenchuang} controlled on $n$ qubits using a number of gates which is linear in $n$. If $U$ is a single-qubit unitary, by conjugating it by these Toffoli gates as is done in \eqref{fig:mcu}, one can implement a version of that unitary controlled on $n$ qubits in a linear number of gates.

\begin{align} \label{fig:mcu}
    \vcenter{\hbox{\scalebox{0.8}{\begin{tikzpicture}
	\begin{pgfonlayer}{nodelayer}
		\node [style=none] (91) at (-3, 2) {};
		\node [style=none] (92) at (4, 2) {};
		\node [style=cnot targ] (93) at (-1, 2) {};
		\node [style=cnot ctrl] (94) at (-1, 0) {};
		\node [style=cnot targ] (95) at (2, 2) {};
		\node [style=cnot ctrl] (96) at (2, 0) {};
		\node [style=cnot ctrl] (97) at (-1, -3) {};
		\node [style=cnot ctrl] (98) at (2, -3) {};
		\node [style=none] (99) at (-3, 0) {};
		\node [style=none] (100) at (4, 0) {};
		\node [style=none] (101) at (-3, -3) {};
		\node [style=none] (102) at (4, -3) {};
		\node [style=none] (103) at (-1, -2) {};
		\node [style=none] (104) at (2, -2) {};
		\node [style=none] (105) at (-1, -1) {};
		\node [style=none] (106) at (2, -1) {};
		\node [style=none] (107) at (-1, -1.25) {$\vdots$};
		\node [style=none] (108) at (2, -1.25) {$\vdots$};
		\node [style=cnot ctrl] (113) at (0.5, 2) {};
		\node [style=none] (114) at (0.5, -4) {};
		\node [style=none] (115) at (0.5, -4.5) {$U$};
		\node [style=none] (118) at (-3, -4.5) {};
		\node [style=none] (119) at (4, -4.5) {};
		\node [style=none] (120) at (0, -5) {};
		\node [style=none] (121) at (0, -4) {};
		\node [style=none] (122) at (1, -4) {};
		\node [style=none] (123) at (1, -5) {};
		\node [style=none] (124) at (-4, 2) {$\ket{0}$};
		\node [style=none] (125) at (8, 2) {};
		\node [style=none] (126) at (15, 2) {};
		\node [style=none] (133) at (8, 0) {};
		\node [style=none] (134) at (15, 0) {};
		\node [style=none] (135) at (8, -3) {};
		\node [style=none] (136) at (15, -3) {};
		\node [style=none] (144) at (11.5, -4) {};
		\node [style=none] (145) at (11.5, -4.5) {$U$};
		\node [style=none] (146) at (8, -4.5) {};
		\node [style=none] (147) at (15, -4.5) {};
		\node [style=none] (148) at (11, -5) {};
		\node [style=none] (149) at (11, -4) {};
		\node [style=none] (150) at (12, -4) {};
		\node [style=none] (151) at (12, -5) {};
		\node [style=none] (152) at (5.75, -1.5) {=};
		\node [style=cnot ctrl] (153) at (11.5, 0) {};
		\node [style=cnot ctrl] (154) at (11.5, -3) {};
		\node [style=none] (155) at (11.5, -2) {};
		\node [style=none] (156) at (11.5, -1) {};
		\node [style=none] (157) at (11.5, -1.25) {$\vdots$};
		\node [style=none] (158) at (7, 2) {$\ket{0}$};
	\end{pgfonlayer}
	\begin{pgfonlayer}{edgelayer}
		\draw (92.center) to (91.center);
		\draw (94) to (93);
		\draw (96) to (95);
		\draw (100.center) to (99.center);
		\draw (102.center) to (101.center);
		\draw (98) to (104.center);
		\draw (97) to (103.center);
		\draw (106.center) to (96);
		\draw (105.center) to (94);
		\draw [style=new edge style 0] (114.center) to (113);
		\draw (119.center) to (118.center);
		\draw [style=new edge style 0] (123.center)
			 to (122.center)
			 to (121.center)
			 to (120.center)
			 to cycle;
		\draw (126.center) to (125.center);
		\draw (134.center) to (133.center);
		\draw (136.center) to (135.center);
		\draw (147.center) to (146.center);
		\draw [style=new edge style 0] (151.center)
			 to (150.center)
			 to (149.center)
			 to (148.center)
			 to cycle;
		\draw (154) to (155.center);
		\draw (156.center) to (153);
		\draw (154) to (144.center);
	\end{pgfonlayer}
\end{tikzpicture}
}}}
\end{align}
\newpage
\subsection{Hyperparameters}
\label{sec:hyperparameters}
\cref{tab:seeds} lists the seed and hyperparameter used for each of our runs, as used by the source code.
\begin{table}[ht!]
\begin{tabular}{lll|lll}
\textbf{Model} & \textbf{Dimension} & \textbf{Seed} & \textbf{Model} & \textbf{Dimension} & \textbf{Seed} \\ \hline
FNet           & 128                & 865026        & Transformer    & 128                & 753912        \\
FNet           & 128                & 896680        & Transformer    & 128                & 740072        \\
FNet           & 128                & 157158        & Transformer    & 128                & 73815         \\
FNet           & 128                & 912324        & Transformer    & 128                & 233689        \\
FNet           & 128                & 336635        & Transformer    & 128                & 791410        \\
FNet           & 128                & 47752         & Transformer    & 128                & 845853        \\
FNet           & 128                & 521512        & Transformer    & 128                & 256815        \\
FNet           & 128                & 410043        & Transformer    & 128                & 616390        \\
FNet           & 128                & 906901        & Transformer    & 128                & 494215        \\
FNet           & 128                & 528631        & Transformer    & 128                & 664970        \\
FNet           & 96                 & 544610        & Transformer    & 96                 & 520839        \\
FNet           & 96                 & 435960        & Transformer    & 96                 & 799108        \\
FNet           & 96                 & 855886        & Transformer    & 96                 & 102749        \\
FNet           & 96                 & 255092        & Transformer    & 96                 & 143407        \\
FNet           & 96                 & 325162        & Transformer    & 96                 & 687589        \\
FNet           & 96                 & 44219         & Transformer    & 96                 & 418518        \\
FNet           & 96                 & 613756        & Transformer    & 96                 & 861454        \\
FNet           & 96                 & 969277        & Transformer    & 96                 & 967387        \\
FNet           & 96                 & 977070        & Transformer    & 96                 & 725501        \\
FNet           & 96                 & 27859         & Transformer    & 96                 & 178408        \\
LSTM           & 128                & 110433        & Quixer         & 6 qbs              & 363500        \\
LSTM           & 128                & 581765        & Quixer         & 6 qbs              & 844031        \\
LSTM           & 128                & 78554         & Quixer         & 6 qbs              & 707858        \\
LSTM           & 128                & 776951        & Quixer         & 6 qbs              & 734571        \\
LSTM           & 128                & 312284        & Quixer         & 6 qbs              & 134017        \\
LSTM           & 128                & 586743        & Quixer         & 6 qbs              & 246154        \\
LSTM           & 128                & 125848        & Quixer         & 6 qbs              & 631481        \\
LSTM           & 128                & 956941        & Quixer         & 6 qbs              & 168344        \\
LSTM           & 128                & 80607         & Quixer         & 6 qbs              & 356044        \\
LSTM           & 128                & 889548        & Quixer         & 6 qbs              & 682772        \\
LSTM           & 96                 & 124900        &                &                    &               \\
LSTM           & 96                 & 951780        &                &                    &               \\
LSTM           & 96                 & 927932        &                &                    &               \\
LSTM           & 96                 & 217285        &                &                    &               \\
LSTM           & 96                 & 627139        &                &                    &               \\
LSTM           & 96                 & 694772        &                &                    &               \\
LSTM           & 96                 & 311736        &                &                    &               \\
LSTM           & 96                 & 939997        &                &                    &               \\
LSTM           & 96                 & 687992        &                &                    &               \\
LSTM           & 96                 & 446349        &                &                    &              
\end{tabular}
\caption{Seeds used for each run for each model, randomly generated using \texttt{torch.randint}.}
\label{tab:seeds}
\end{table}

\end{document}